\documentclass[fleqn]{article}
\title{Generalized Unrelated Machine Scheduling Problem} 

\author{
Shichuan Deng\thanks{IIIS, Tsinghua University, China. 
Email: \texttt{dsc15@mails.tsinghua.edu.cn}
} 
\and Jian Li\thanks{IIIS, Tsinghua University, China. 
Email: \texttt{lijian83@mail.tsinghua.edu.cn}
} 
\and Yuval Rabani\thanks{The Hebrew University of Jerusalem, Israel. 
Email: \texttt{yrabani@cs.huji.ac.il}
}}

\usepackage{amsmath}
\usepackage{amsfonts}
\usepackage{bbm,bm}
\usepackage{caption,subcaption}
\usepackage{amssymb}
\usepackage{graphicx,url}
\usepackage{multirow,array}
\usepackage{makecell}

\usepackage[linesnumbered,ruled,vlined]{algorithm2e}
\usepackage[margin=1in]{geometry}
\usepackage{xspace}
\usepackage[table,dvipsnames]{xcolor}
\usepackage{amsthm}
\usepackage{enumitem}
\usepackage{hyperref}
\hypersetup{
    colorlinks=true,
    linkcolor=black,
    urlcolor=cyan,
    citecolor=ForestGreen
}
\usepackage[capitalise]{cleveref}

\DeclareMathOperator*{\argmax}{arg\,max}

\newcommand*{\E}{\mathbb{E}}

\newcommand{\ceil}[1]{\left\lceil #1 \right\rceil}
\newcommand{\angles}[1]{[ #1 ]}
\newcommand{\assigned}[1]{[ #1 ]}

\newcommand{\da}{\downarrow}
\newcommand{\mcal}[1]{{\mathcal{#1}}}

\newcommand{\prevp}[1]{\mathsf{prev}(#1)}
\newcommand{\nextp}[1]{\mathsf{next}(#1)}

\newcommand{\boldp}{\bm{p}}
\newcommand{\boldu}{\bm{u}}
\newcommand{\boldv}{\bm{v}}
\newcommand{\boldz}{\bm{z}}

\newcommand{\jobs}{\mcal{J}}
\newcommand{\machines}{\mcal{M}}
\newcommand{\load}{\mathsf{load}}
\newcommand{\makespan}{\Phi}
\newcommand{\innorm}{\psi}
\newcommand{\outnorm}{\phi}

\newcommand{\pos}{\mathsf{POS}}
\newcommand{\pwr}{\mathsf{PWR}}
\newcommand{\topp}[1]{\operatorname{Top}_{#1}}
\newcommand{\topl}[2]{\operatorname{Top}_{#1}\left(#2\right)}
\newcommand{\smn}{{\scshape SymMonNorm}\xspace}

\newcommand{\primal}{\operatorname{P-LB}}
\newcommand{\dual}{\operatorname{D-LB}}
\newcommand{\anotherdual}{\operatorname{Q}}
\newcommand{\thresholds}{\bm{R}}
\newcommand{\vco}{\bm{o}}
\newcommand{\vcv}{\bm{v}}
\newcommand{\opt}{\mathsf{opt}}
\newcommand{\rexp}{\vec{\varrho}}
\newcommand{\rexpsub}[1]{\vec{\varrho}_{#1}}

\newcommand{\R}{\mathbb{R}}
\newcommand{\Rpos}{\mathbb{R}_{\geq0}}
\newcommand{\Zpos}{\mathbb{Z}_{\geq0}}

\newcommand{\etalcite}[1]{\textit{et~al.}~\cite{#1}}

\newcommand{\mnlb}{\textsf{MinNormLB}\xspace}
\newcommand{\glb}{\textsf{GLB}\xspace}
\newcommand{\maxtopk}{\textsf{GLB-MaxTopK}\xspace}
\newcommand{\normval}{\textsf{NormLin}\xspace}

\newtheorem{theorem}{Theorem}
\newtheorem{lemma}[theorem]{Lemma}

\newtheorem{claim}[theorem]{Claim}
\newtheorem{observation}[theorem]{Observation}

\theoremstyle{definition}
\newtheorem{definition}[theorem]{Definition}
\theoremstyle{remark}
\newtheorem*{remark*}{Remark}
\newtheorem*{note*}{Note}

\begin{document}
\maketitle

\begin{abstract}
	We study the \emph{generalized load-balancing} (\glb) problem, where we are given $n$ jobs, each of which needs to be assigned to one of $m$ unrelated machines with processing times $\{p_{ij}\}$. 
	Under a job assignment $\sigma$, the \emph{load} of each machine $i$ is $\innorm_i(\boldp_{i}\assigned{\sigma})$ where $\innorm_i:\R^n\rightarrow\Rpos$ is a symmetric monotone norm
	and $\boldp_{i}\assigned{\sigma}$ is the $n$-dimensional vector $\{p_{ij}\cdot \mathbf{1}[\sigma(j)=i]\}_{j\in [n]}$.
    Our goal is to minimize the \emph{generalized makespan} $\outnorm(\load(\sigma))$, where $\outnorm:\R^m\rightarrow\Rpos$ is another symmetric monotone norm and $\load(\sigma)$ is the $m$-dimensional machine load vector. 
    This problem significantly generalizes many classic optimization problems, e.g., makespan minimization, set cover, minimum-norm load-balancing, etc.
	
	We obtain a polynomial time randomized algorithm that achieves an approximation factor of $O(\log n)$, matching the lower bound of set cover up to constant factor. We achieve this by rounding a novel \emph{configuration} LP relaxation
	with exponential number of variables. To approximately solve the configuration LP, we design an approximate separation oracle for its dual program. In particular, the separation oracle can be reduced to the \emph{norm minimization with a linear constraint} (\normval) problem and we devise a polynomial time approximation scheme (PTAS) for it, which may be of independent interest.
\end{abstract}

\section{Introduction}

In the \emph{generalized load-balancing} (\glb) problem, we are given a set $\machines$ of $m$ unrelated machines, a set $\jobs$ of $n$ jobs.
Let $p_{ij}>0$ be the processing time of job $j\in \jobs$ on machine $i\in \machines$, and we denote $\boldp_i=\{p_{ij}\}_{j\in \jobs}$.
Suppose $\sigma:\jobs\rightarrow\machines$ is an assignment of all jobs to machines.
Under assignment $\sigma$, we use  
$\boldp_{i}\assigned{\sigma}$ to denote the $n$-dimensional vector 
$\{p_{ij}\cdot \mathbf{1}[\sigma(j)=i]\}_{j\in \jobs}$ (i.e., we zero out
all entries of $\boldp_i$ that are not assigned to $i$).
The \emph{load} of machine $i$ is define to be $\load_{i}(\sigma)=\innorm_i(\boldp_{i}\assigned{\sigma})$, where $\innorm_i:\R^n\rightarrow\Rpos$ is a symmetric monotone norm\footnote{Recall $\innorm:\R^\mcal{X}\rightarrow\Rpos$ is a norm if: 
(i) $\innorm(\boldu)=0$ if and only if $\boldu=0$, 
(ii) $\innorm(\boldu+\boldv)\leq\innorm(\boldu)+\innorm(\boldv)$ for all $\boldu,\boldv\in\R^\mcal{X}$, 
(iii) $\innorm(\theta\boldu)=|\theta|\innorm(\boldu)$ for all $\boldu\in\R^\mcal{X},\theta\in\mathbb{R}$. A norm $\innorm$ is monotone if $\innorm(\boldu)\leq\innorm(\boldv)$ for all $0\leq\boldu\leq\boldv$, and symmetric if $\innorm(\boldu)=\innorm(\boldu')$ for any permutation $\boldu'$ of $\boldu$.} for each $i\in\machines$, called the \emph{inner norms}.
We define the \emph{generalized makespan} of assignment $\sigma$ as 
$\makespan(\sigma)=\outnorm(\load(\sigma))$
where $\load(\sigma)=\{\load_{i}(\sigma)\}_{i\in\machines}\in\Rpos^m$ (called the load vector) and $\outnorm:\R^m\rightarrow\Rpos$ is another symmetric monotone norm, called the \emph{outer norm}.
Our goal is to find an assignment $\sigma$ such that the generalized makespan $\makespan(\sigma)$ is minimized.

Many special cases of \glb have been studied in the setting of unrelated machine scheduling. 
The classic makespan minimization problem is one such example, where $\outnorm=\mcal{L}_\infty$ and $\innorm_i=\mcal{L}_1$ for each $i\in\machines$. 
For any $\alpha<3/2$, there is no $\alpha$-approximation algorithm for makespan minimization unless $\operatorname{P=NP}$ \cite{lenstra1990approximation}, while improving the current best approximation ratio of 2 \cite{lenstra1990approximation,shmoys1993approximation} remains a longstanding open problem.
Several variants of makespan minimization with the objective being the general $\mcal{L}_p$ norm of the load vector
(i.e., $\outnorm=\mcal{L}_p$ and $\innorm_i=\mcal{L}_1$)
have also been studied extensively.
Constant approximations are known  \cite{alon1998scheduling,azar2005convex, kumar2009unified,makarychev2018optimization}, with better-than-two factors depending on $p$.
Recently, Chakrabarty and Swamy \cite{chakrabarty2019approximation} proposed the minimum-norm load-balancing (\mnlb) problem, which significantly generalizes makespan minimization and its $\mcal{L}_p$ norm variants. 
In \mnlb the outer norm can be a general symmetric monotone norm, but each inner norm is still $\innorm_i=\mcal{L}_1$.
They provided an LP-based constant factor approximation algorithm. 
The approximation factor was subsequently improved to $(4+\epsilon)$ in \cite{chakrabarty2019simpler}, and the current best result is a $(2+\epsilon)$-approximation \cite{ibrahimpur2021minimum}, almost matching the best-known guarantee for makespan minimization.

Although the outer norm has been studied extensively, the inner
norm has received relatively less attention except the $\mcal{L}_1$ case.
However, other inner norms can be naturally motivated.
One example is the $\mcal{L}_\infty$ norm, which can model the RAM size lower bound for each machine, if we view the job sizes $\{p_{ij}\}$ as RAM size requirements and process the assigned jobs sequentially.
Other $\mcal{L}_p$ norms may be used for minimizing average latency \cite{azar2005convex}, energy efficient scheduling \cite{makarychev2018optimization}, etc.

Before stating our results for \glb, we discuss a line of closely related problems.
Here the load of machine $i$ is defined by a submodular set function mapping the assigned jobs $\sigma^{-1}(i)$ to a scalar.
Svitkina and Fleischer \cite{svitkina2011submodular} study this problem where the outer norm is $\outnorm=\mcal{L}_\infty$ and show that there is no polynomial time algorithm that can achieve an approximation ratio of $o(\sqrt{n/\log n})$. They also provide a factor $O(\sqrt{n\log n})$ approximation algorithm. 
When $\outnorm=\mcal{L}_1$, the problem is known as the minimum submodular-cost allocation problem \cite{chekuri2011submodular}, and admits a tight $O(\log n)$-approximation, with set cover being its special case \cite{svitkina2010facility}.
This result implies a nontrivial approximation for the following setting of \glb.
Let $\topp{k}$ be the top-$k$ norm defined as follows: for any vector $\boldu \geq 0$, 
$\topl{k}{\boldu}$ outputs the sum of the $k$ largest entries in $\boldu$.
For each $J\subseteq\jobs$, define $\boldp_i\angles{J}=\{p_{ij}\cdot\mathbf{1}[j\in J]\}_{j\in\jobs}$.
It is easy to verify that $f(J)=\topl{k_i}{\boldp_i\angles{J}}$ is a submodular set function for each $k_i\in[n]$ and $\boldp_i\geq0$. Thus, \glb with $\outnorm=\mcal{L}_1$ and $\innorm_i=\topp{k_i}$ can be reduced to the minimum submodular-cost allocation problem and 
we readily obtain an $O(\log n)$-approximation using the result in \cite{svitkina2010facility}.

However, the reduction above does not work for general symmetric monotone inner norm $\innorm_i$,
since a general symmetric monotone norm does not necessarily induce a submodular set function. 
For example, let $\innorm_i(\boldu)=\max\{2\topl{1}{\boldu},\topl{3}{\boldu}\}$, and $\boldp_i=\mathbf{1}_n$ (all processing time $p_{ij}=1$).
It is easy to see that $\innorm_i$ is indeed a symmetric monotone norm.
However, assigning job-sets $\{j_1\}$, $\{j_1,j\}$, $\{j_1,j_2\}$, $\{j_1,j_2,j\}$ incurs loads $2,2,2,3$, respectively. The property of diminishing marginal returns (see, e.g., \cite{edmonds2001submodular}) is violated and the function $f(J)=\innorm_i(\boldp_i\angles{J})$ is not submodular.

\subsection{Our Contributions}

Motivated by the recent progress on unrelated machine scheduling problems with more general objectives, we  study \glb systematically with different combinations of inner norms and outer norms.
A summary of known and new results can be found in Figure~\ref{figure:result:table}.

The first natural question is whether it is possible to obtain constant factor approximation algorithms for 
\glb with general outer and inner norms.
However, it is not difficult to see an $\Omega(\log n)$ lower bound for \glb, even when $\innorm_i=\mcal{L}_\infty$
and $\outnorm$ is simply the $\mcal{L}_1$ norm, via reduction from the unweighted set cover problem.
Formally, in an unweighted set cover instance, we are given a family of $m$ subsets $\mcal{S}=\{S_1,S_2,\dots,S_m\}$ of $[n]$, and the goal to find $I\subseteq[m]$ with minimum cardinality such that $\bigcup_{i\in I}S_i=[n]$.
We identify $m$ machines with $[m]$ and $n$ jobs with $[n]$. 
For each machine $i\in[m]$ and job $j\in[n]$, define the job size $p_{ij}=1$ if $j\in S_i$, and $\infty$ if $j\notin S_i$. The inner norms are $\innorm_i=\mcal{L}_\infty$ for each $i\in\machines$, and the outer norm $\outnorm$ is simply the $\mcal{L}_1$ norm. 
It is easy to verify that, any assignment $\sigma$ with a \emph{finite} generalized makespan $\makespan(\sigma)$ corresponds to a solution $I_\sigma=\{i\in[m]:\sigma^{-1}(i)\neq\emptyset\}$ to the set cover instance with the same objective, and vice versa. 
Using the $\operatorname{NP}$-hardness of approximating set cover \cite{dinur2014analytical}, we obtain the following theorem.

\begin{theorem}
\label{theorem:lower:bound:glb}
For every fixed constant $\epsilon>0$, it is $\operatorname{NP}$-hard to approximate \emph{\glb} within a factor of $(1-\epsilon)\ln n$, even when $\outnorm=\mcal{L}_1$ and $\innorm_i=\mcal{L}_\infty$ for each $i\in\machines$.
\end{theorem}

On the positive side, we show, somewhat surprisingly, that \glb with general inner and outer norms,
admits an $O(\log n)$ factor approximation algorithm, almost matching the approximability of set cover,
a very special case of \glb. This is the main technical result of this paper.

\begin{theorem}\label{theorem:main:normnorm}
There exists a polynomial time randomized algorithm for \emph{\glb} that, with high probability, outputs an $O(\log n)$-approximate solution.
\end{theorem}

We also study another special case.
In particular, we obtain an LP-based $(3+\epsilon)$-approximation for the special case where the outer norm $\outnorm$ is $\mcal{L}_\infty$, and the inner norm is $\innorm_i=\topp{k_i},\,k_i\in[n]$ for each machine $i\in\machines$. 
We provide details of this algorithm (\cref{theorem:maxoftopl}) in the appendix.


\begin{figure}[t]
    \centering
    {\renewcommand{\arraystretch}{1.2}
    \begin{tabular}
    {|l|l|l|l|}
 \hline
 \multirow{2}{*}{
 	\makecell[l]{\textbf{Machine load} \\ 
 	\textbf{(inner norm $\innorm_i$)}}} 
 & \multicolumn{3}{c|}{
 	\textbf{Generalized makespan (outer norm $\outnorm$)}}\\
 \cline{2-4}
 & $\mcal{L}_1$ (SUM) & $\mcal{L}_\infty$ (MAX) & \smn\\
 \hline\hline
\multirow{2}{*}{$\mcal{L}_1$ (SUM)} & \multirow{2}{*}{\textit{trivial}} & 
\multirow{2}{*}{\makecell[cl]{
LB: 3/2 \cite{lenstra1990approximation};\\
UB: 2 \cite{lenstra1990approximation,shmoys1993approximation}}} & 
\multirow{2}{*}{\makecell[cl]{
LB: 3/2 \cite{lenstra1990approximation};\\
UB: $2+\epsilon$ \cite{ibrahimpur2021minimum}}}\\
&&&\\
 \hline
$\mcal{L}_\infty$ (MAX) & \multirow{3}{*}{\makecell[cl]{LB: $(1-\epsilon)\ln n$\\
\quad(\cref{theorem:lower:bound:glb});\\
UB: $O(\log n)$ \cite{svitkina2010facility}}} 
 & \textit{trivial} 
 & \multirow{6}{*}{\makecell[cl]{
 LB: $(1-\epsilon)\ln n$\\
 \quad(\cref{theorem:lower:bound:glb});\\
 \\
 UB: $O(\log n)$\\ \quad{\bfseries(\cref{theorem:main:normnorm})}}} \\
 \cline{1-1}\cline{3-3}
 \multirow{2}{*}{$\operatorname{Top}_{k_i}$ norm} &  & \multirow{2}{*}{\makecell[cl]{
 LB: 3/2 \cite{lenstra1990approximation}; UB: $3+\epsilon$\\ \quad{\bfseries(\cref{theorem:maxoftopl})}}} & \\
 & & &\\
 \cline{1-3}
 \multirow{3}{*}{\smn} & \multirow{2}{*}{\makecell[cl]{LB: $(1-\epsilon)\ln n$\\
\quad(\cref{theorem:lower:bound:glb});}} &
\multirow{2}{*}{\makecell[cl]{LB: 3/2 \cite{lenstra1990approximation};}} & \\
 &&&\\
 \cline{2-3}
 &\multicolumn{2}{l|}{\makecell[cl]{UB: $O(\log n)$ {\bfseries(\cref{theorem:main:normnorm})}}}&\\
 \hline
\end{tabular}}
    \caption{A summary of \glb cases with different objectives, where $n$ is the number of given jobs.
    LB (lower bound) represents known inapproximability results, and UB (upper bound) represents current approximation guarantees.
    \smn refers to symmetric monotone norms.}
    \label{figure:result:table}
\end{figure}

\subparagraph{Overview of our algorithm for Theorem \ref{theorem:main:normnorm}.}
We devise a randomized rounding algorithm based on a novel configuration LP. 
A configuration is an arbitrary set of jobs.
We create variables $x_{i,J}\in[0,1]$ for each $i\in\machines$ and $J\subseteq\jobs$, indicating whether $J$ is the configuration of $i$. 
Our goal is to obtain a fractional solution $x\in[0,1]^{\machines\times2^\jobs}$, such that its generalized makespan is bounded by some unknown optimal solution.
To this end, using a standard technique (see, e.g., \cite{chakrabarty2019approximation,ibrahimpur2021minimum}), we guess a logarithmic number of constraints, also referred to as ``budget constraints'', and guarantee that at least one such guess produces a feasible LP that leads to a good approximate solution.
We also add various ingredients to facilitate the later rounding procedure, by restricting the support size of vertex solutions, and including parameters to tighten and relax the constraints.

Though the number of variables is exponential, there are only a polynomial number of constraints, thus we consider solving its dual.
Unfortunately, the dual constraints are difficult to separate exactly, so we consider approximate separation on \emph{another} parameterized linear program $\operatorname{Q}$, and employ the following round-or-cut framework, which has recently been a powerful tool in the design of approximation algorithms (see, e.g., \cite{anegg2020technique,carr2000strengthen,chakrabarty2018generalized}): Suppose given any candidate dual solution $y$, we either find a violated constraint of $\operatorname{Q}$, or certify that another differently parameterized program $\operatorname{Q}'$ is feasible, then the ellipsoid algorithm, in polynomial time, either concludes that $\operatorname{Q}$ is infeasible, or that $\operatorname{Q}'$ is feasible. By establishing an equivalence between the feasibility of $\operatorname{Q}$ (also $\operatorname{Q}'$) and the original relaxation, and enumerating all possible $\operatorname{Q}$, we can solve the primal relaxation up to a constant accuracy.

We use randomized rounding on the primal solution similar to weighted set cover (see, e.g., \cite{vazirani2001approximation}), and assign all jobs with high probability. 
At this junction, we isolate two factors that affect the final approximation guarantee. 
The first arises from assigning multiple distinct configurations to the same machine, incurring a proportional approximation factor. 
The second comes from the violation of budget constraints in the primal relaxation by the integral solution, and the approximation ratio is proportional to the maximum factor of violation. 
These two factors add up to $O(\log n)$ in the rounding analysis, yielding our main result.

Finally, we highlight a useful subroutine for the so-called \emph{norm minimization with a linear constraint} (\normval) problem, which we employ in the approximate separation oracle. 
Roughly speaking, \normval generalizes the min-cost version of the knapsack problem with a value lower bound, using a symmetric monotone norm objective (see \cref{section:oracle} for the precise definition). We devise a polynomial time approximation scheme (PTAS) for \normval. To the best of our knowledge, this is the first known approximation algorithm for this problem, which may be useful in other context.
Svitkina and Fleischer \cite{svitkina2011submodular} consider a similar problem, where the objective is a submodular set function. They show a lower bound of $\Omega(\sqrt{n/\log n})$ (for any bi-criteria approximation), even for monotone submodular functions and 0-1 values.



\subsection{Related Work}

In a closely related \emph{simultaneous optimization} problem, one seeks a job assignment incurring a load vector that simultaneously approximates all optimums under different \emph{outer} norms. 
Though no simultaneous $\alpha$-approximate solutions exist for unrelated machines for any constant $\alpha$, even for simple $\innorm_i=\mcal{L}_1$ inner norms \cite{azar2004all}, Alon \etalcite{alon1998scheduling} give an algorithm in the case of restricted assignment (i.e., each job has a fixed size but can only be assigned to a subset of machines) and unit-size jobs, which is simultaneously optimal for all $\mcal{L}_p$ norms.
Azar \etalcite{azar2004all} extend this result to a simultaneous 2-approximation for all $\mcal{L}_p$ norms under restricted assignment.
This is generalized to a simultaneous 2-approximation (again in the restricted assignment setting) for all symmetric monotone norms by Goel and Meyerson \cite{goel2006simultaneous}.

Another relevant result is the generalized machine activation problem introduced by Li and Khuller \cite{li2011generalized}, where a machine activation cost is incurred for each machine $i$, by applying a non-decreasing, piece-wise linear function $\omega_i:\Rpos\rightarrow\Rpos$ on the sum of job sizes assigned to $i$. They achieve a bi-criteria approximation guarantee on fractional solutions, and obtain various almost-tight application results.

The configuration LP is used in many allocation/assignment optimization problems, for example, the Santa Claus problem and fair allocation of indivisible goods \cite{asadpour2012santa,asadpour2010approximation,bansal2006santa,feige2008allocations}, and restricted scheduling on unrelated machines \cite{jansen2017configuration,svensson2012santa}.

\subsection{Organization}

We start by stating some notations and preliminaries in \cref{section:preliminary}. 
In \cref{section:normnorm}, we present our main LP-rounding algorithm, and the proof of \cref{theorem:main:normnorm}. We provide the details of the approximate separation oracle in \cref{section:oracle}. Finally, we provide a constant factor approximation algorithm for a special case of \glb in the appendix.

\section{Preliminaries}\label{section:preliminary}

Throughout this paper, for vector $\boldu\in\Rpos^\mcal{X}$, define $\boldu^\da$ as the non-increasingly sorted version of $\boldu$, and $\boldu\angles{S}=\{\boldu_j\cdot\mathbf{1}[j\in S]\}_{j\in\mcal{X}}$ for each $S\subseteq\mcal{X}$.
Let $\topp{k}:\R^\mcal{X}\rightarrow\Rpos$ be the top-$k$ norm that returns the sum of the $k$ largest \emph{absolute values} of entries in any vector, $k\leq |\mcal{X}|$.
Denote $[n]$ as the set of positive integers no larger than $n\in\mathbb{Z}$, and $a^+=\max\{a,0\},\,a\in\mathbb{R}$. 

\begin{claim}\label{claim:topl}
\emph{(\cite{ibrahimpur2021minimum}).}
For each $n$-dimensional vector $\boldu\geq0$ and $k\in[n]$, one has \[\topl{k}{\boldu}=\min_{t\geq0}\Big\{kt+\sum_{j\in[n]}(\boldu_j-t)^+\Big\}
=k\boldu^\da_k+\sum_{j\in[n]}(\boldu_j-\boldu^\da_k)^+,\]
i.e., the minimum is attained at the $k$-th largest entry of $\boldu$.
\end{claim}

The following lemma is due to Goel and Meyerson \cite{goel2006simultaneous}.

\begin{lemma}\label{lemma:majorization}
\emph{(\cite{goel2006simultaneous}).}
	If $\boldu,\boldv\in\Rpos^\mcal{X}$ and $\alpha\geq0$ satisfy $\topl{k}{\boldu}\leq\alpha\cdot\topl{k}{\boldv}$ for each $k\leq|\mcal{X}|$, one has $\innorm(\boldu)\leq\alpha\cdot \innorm(\boldv)$ for any symmetric monotone norm $\innorm:\R^\mcal{X}\rightarrow\Rpos$.
\end{lemma}

We need the following Chernoff bounds.

\begin{lemma}\label{lemma:chernoff}
\emph{(Chernoff bounds (see, e.g., \cite{mitzenmacher2005probability})).}
Let $X_1,\dots,X_n$ be independent Bernoulli variables with $\E[X_i]=p_i$. 
Let $X=\sum_{i=1}^nX_i$ and $\mu=\E[X]=\sum_{i=1}^np_i$. 
For $\nu\geq6\mu$, one has
\(
\Pr[X\geq\nu]\leq 2^{-\nu}
\).
\end{lemma}

\section{The Generalized Load Balancing Problem}\label{section:normnorm}

In this section, we study \glb and prove \cref{theorem:main:normnorm}.
For clarity of presentation, we have not optimized the constant factors.

\subsection{The Configuration LP}\label{section:main:lp}

Our algorithm is based on rounding polynomial many configuration LPs. 
In each configuration LP, instead of using the natural $x_{ij}$ variables to indicate the job assignment, e.g., \cite{lenstra1990approximation,shmoys1993approximation}, 
we create variables $x_{i,J}\in[0,1]$ for each $i\in\machines$ and $J\subseteq\jobs$, indicating whether $J$ is the set of jobs assigned to $i$.
Similar configuration variables have been used in other allocation/assignment optimization problems \cite{asadpour2012santa,asadpour2010approximation,bansal2006santa,feige2008allocations,jansen2017configuration,svensson2012santa}.


Fix an unknown optimal assignment $\sigma^\star:\jobs\rightarrow\machines$. Let $\vco=\load(\sigma^\star)\in\Rpos^m$ be the optimal load vector and $\opt=\outnorm(\vco)$ be the optimal objective thereof.
We need to guess certain characteristics of the optimal solution (we will show there are polynomial many possible guesses), and for each guess we write a configuration LP \ref{lp:normnorm:primal}. Here $\thresholds,\lambda,\tau$ are the parameters we need to adjust, and we explain them later.

We first explain constraints \eqref{lp:normnorm:primal1}, which can be used to bound the $\outnorm$-norm. 
Using \cref{lemma:majorization}, to obtain an $O(\log n)$-approximate solution, it suffices to obtain a load vector $\load(\sigma)\in\Rpos^m$ such that $\topl{k}{\load(\sigma)}$ is bounded by $O(\log n)\cdot\topl{k}{\vco}$ for each $k\in[m]$. 
Rather than bounding the top-$k$ norms for each $k$, we use the trick developed in \cite{chakrabarty2019approximation,ibrahimpur2021minimum}, focusing on a subset of geometrically-placed indexes of $[m]$, e.g., all integer powers of 2 smaller than $m$. Let $\pos\subseteq[m]$ be this index subset (which we define formally later). Our LP seeks a fractional solution that (roughly) has $\topp{k}$ norm at most $\topl{k}{\vco}$ for each $k\in\pos$.

\begin{align*} 
	\text{min} && 0\tag{$\primal(\thresholds,\lambda,\tau)$}\label{lp:normnorm:primal}\\
	\text{s.t.} && k\rho_k+\sum_{i\in\machines,J\subseteq\jobs}\left(h(\innorm_i(J)/\tau)-\rho_k\right)^+x_{i,J}&\leq\topl{k}{\rexp}\quad\forall k\in\pos
	\tag{$\primal.1$}\label{lp:normnorm:primal1}\\
	&& \sum_{J\subseteq\jobs}x_{i,J} &\leq1 \quad\forall i\in\machines\tag{$\primal.2$}\label{lp:normnorm:primal2}\\
	&& \sum_{i\in\machines,J\ni j}x_{i,J} &\geq \lambda \quad\forall j\in\jobs\tag{$\primal.3$}\label{lp:normnorm:primal3}\\
	&& \sum_{(i,J):\innorm_i(J)>\tau}x_{i,J}&\leq0\tag{$\primal.4$}\label{lp:normnorm:primal4}\\
	&& \sum_{i\in\machines,J\subseteq\jobs}x_{i,J}&\leq n\tag{$\primal.5$}\label{lp:normnorm:primal5}\\
	&& x&\geq0.
\end{align*}%

To write linear constraints on such $\topp{k}$ norms, we utilize \cref{claim:topl} and work with close estimates of $\vco^\da_k,\,k\in\pos$, denoted by vector $\thresholds=\{\rho_k\}_{k\in\pos}\in\Rpos^\pos$ (formally defined later). 
With a slight abuse of notation, denote $\innorm_i(J)=\innorm_i(\boldp_{i}\angles{J})$ as the $\innorm_i$-norm of assigning job-set $J$ to machine $i$. 
We use $k\rho_k+\sum_{i\in\machines,J\subseteq\jobs}(\innorm_i(J)-\rho_k)^+x_{i,J}$ to represent the $\topp{k}$ norm of LP solution $x$. 
We also use a certain ``expansion'' $\rexp\in\Rpos^m$ of $\{\rho_k\}_{k\in\pos}$ (proposed in \cite{ibrahimpur2021minimum}) as an upper bound vector, as shown in constraint \eqref{lp:normnorm:primal1}.

We need another important ingredient, that is to tighten the above-mentioned $\topp{k}$ norm constraints by rounding up each $\innorm_i(J)$ to the nearest value in $\thresholds$. Let $h:\Rpos\rightarrow\Rpos$ be this round-up function determined by $\thresholds$ (See \cref{figure:p:h}). 
Though this may create an unbounded gap between $h(\innorm_i(J))$ and $\innorm_i(J)$, the new LP is still feasible under suitable parameters.
This rounding-up process is technically useful in our final analysis of randomized rounding.

The configuration LP is parameterized by $\thresholds$ (therefore $\rexp$ and $h$ are also determined), a constant $0\leq \lambda\leq 1$ and a parameter $\tau\geq 1$. 
As explained above, the constraints \eqref{lp:normnorm:primal1} are for bounding $\topp{k}$ norms. Note that each $\innorm_i(J)$ is scaled by $\tau^{-1}\leq1$, due to our approximate procedure of solving the LP (see \cref{section:main:rounding}).
\eqref{lp:normnorm:primal2} says each machine can be selected to an extent of at most 1.
\eqref{lp:normnorm:primal3} says each job needs to be assigned to an extent of at least $\lambda\leq1$, relaxed again because of our approximate procedure. 
In \eqref{lp:normnorm:primal4}, each assignment $(i,J)$ that incurs a cost larger than $\tau$ is set to 0 (similar to the parametric pruning trick used in classic unrelated machine scheduling problems \cite{lenstra1990approximation,shmoys1993approximation}).
In \eqref{lp:normnorm:primal5}, the total extent of assignment over all possible $(i,J)$ is at most $n$, because there are only $n$ jobs. This constraint, together with \eqref{lp:normnorm:primal2}, further limits the support size of vertex solutions, which will be crucial for optimizing our approximation ratio. 
We also have the obvious constraints $x\geq0$.

\subparagraph*{The guessing step.}
We need some additional notations.
Let $\pos=\{\min\{2^s,m\}:s\in\Zpos\}$. We have $1,m\in\pos$, $\pos\subseteq[m]$ and $|\pos|\leq\log_2m+2$. For each $k\in[m]$, define $\nextp{k}$ as the smallest number in $\pos$ that is larger than $k$, and $\nextp{m}=m+1$; $\prevp{k}$ as the largest number in $\pos$ that is smaller than $k$, and $\prevp{1}=0$. 

Recall $\sigma^\star$ is an optimal assignment.
We guess the machine $i^\star$ with the largest $\load_{i^\star}(\sigma^\star)$, and the job $j^\star=\argmax_{\sigma^\star(j)=i^\star}p_{i^\star j}$
(i.e., $j^\star$ is the job assigned to machine $i^\star$ with largest processing time).
Assume w.l.o.g. that $\innorm_{i^\star}(\{j^\star\})=1/n$. 
Since the norms are monotone, one has $\load_{i}(\sigma^\star)\leq \innorm_{i^\star}(p_{i^\star j^\star}\cdot\mathbf{1}_n)\leq n\cdot \innorm_{i^\star}(p_{i^\star j^\star}\cdot e_{j^\star})=n\cdot \innorm_{i^\star}(\{j^\star\})=1$ for each $i\in\machines$, 
where $e_{j^\star}$ is the zero vector except that the $j^\star$-th coordinate is $1$.
Thus, we can assume $\vco\in[0,1]^m$ in the following discussion. In the remainder of this section, we assume $i^\star$ and $j^\star$ are correctly guessed (it is easy to see there are only polynomial number of possible guesses).

We now guess a \emph{non-increasing} vector $\thresholds=\{\rho_k\}_{k\in\pos}\in\Rpos^\pos$, where
\begin{enumerate}[label=(\roman*)]
	\item each $\rho_k$ is an integer power of 2 in $[1/(2mn),1]$,
	\item there are at most $\log_2 n+4$ distinct entries in $\thresholds$.
\end{enumerate}
The length is $|\pos|=O(\log m)$ and the number of possible values is the number of integer powers of 2 in $[1/(2mn),1]$, which is $O(\log (mn))$. 
The number of such possible non-increasing vectors is at most the number of natural number (i.e., $\mathbb{Z}_{\geq0}$) solutions to the simple equation $n_1+n_2+\cdots+n_{O(\log (mn))}=O(\log m)$, thus $\binom{O(\log (mn))}{O(\log m)}\leq2^{O(\log (mn))}=(mn)^{O(1)}$. 

We do not make any additional assumptions on $\thresholds$ for now.
For each $\thresholds$, we define a \emph{non-increasing} vector $\rexp\in\Rpos^m$ as its expansion (similar to \cite{ibrahimpur2021minimum}), and a \emph{non-decreasing} function $h:\Rpos\rightarrow\Rpos$ satisfying $h(x)\geq x$ for $x\geq0$ (see \cref{figure:p:h} for an example), where
\begin{equation*}
\rexpsub{k}=\left\{\begin{array}{ll}
\rho_k & k\in\pos\\
\rho_{\prevp{k}} & k\notin\pos,
\end{array}
\right.
\quad
h(x)=\left\{\begin{array}{cl}
\min\{t\in \thresholds:t\geq x\}	 & x\leq\max\{t\in \thresholds\}\\
1 & x\in(\max\{t\in \thresholds\},1] \\
x & x>1.
\end{array}
\right.
\end{equation*}

\begin{figure}[t]
    \begin{subfigure}[t]{0.5\textwidth}
    \centering
    \includegraphics[width=\textwidth]{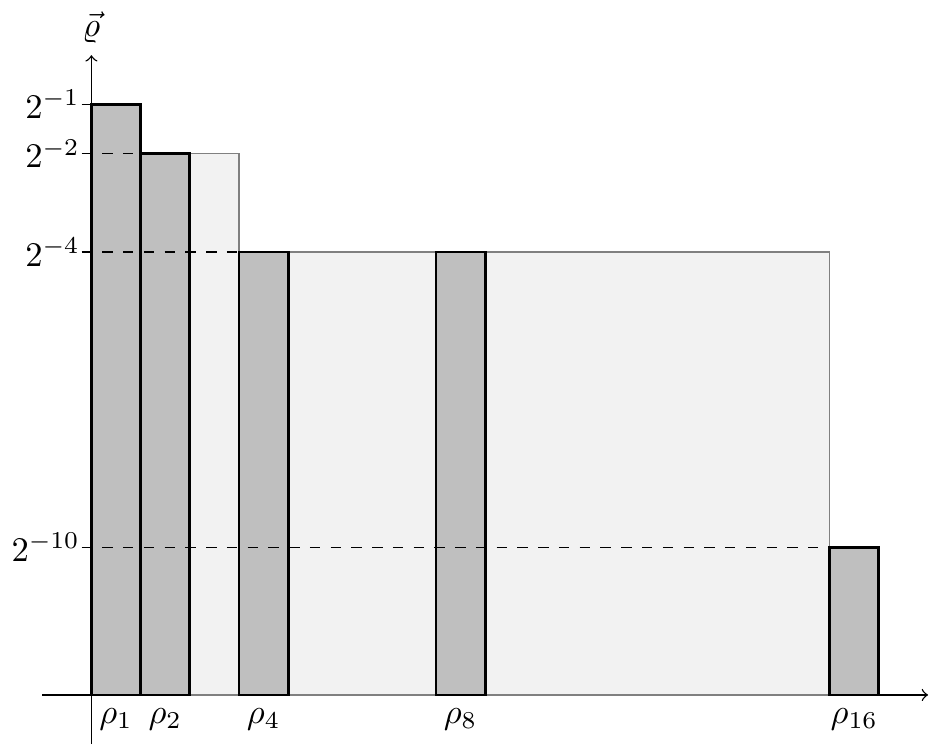}
    \subcaption{$\rexp\in\Rpos^{16}$, where the entries are non-increasing and represented using rectangles with unit width.
    }
    \label{subfigure:p}
    \end{subfigure}
    \hfill
    \begin{subfigure}[t]{0.42\textwidth}
    \centering
    \includegraphics[width=\textwidth]{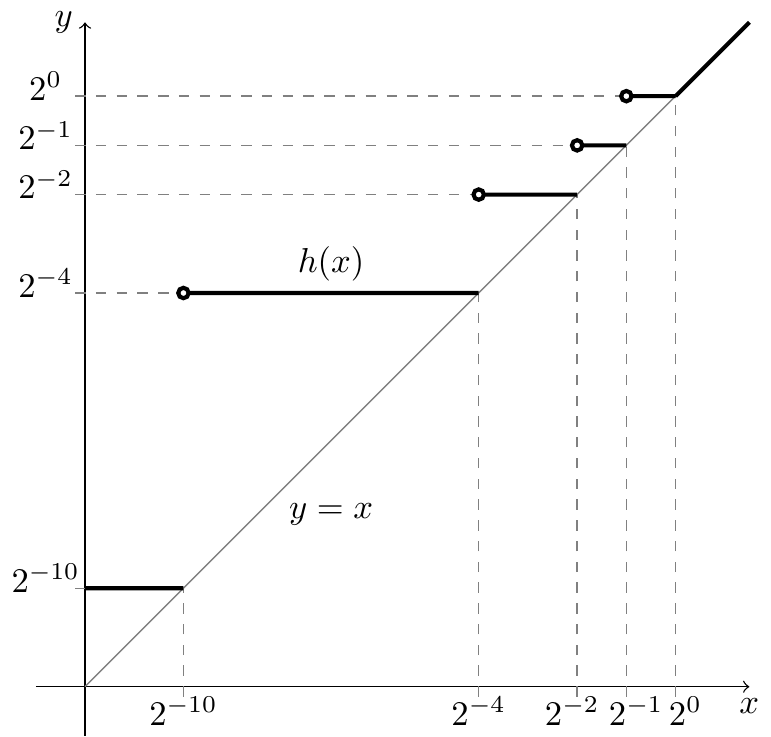}
    \subcaption{The round-up function $h$. The axes are both in logarithmic scale.}
    \label{subfigure:h}
    \end{subfigure}
    \caption{An example when $m=16$ and $\pos=\{1,2,4,8,16\}$. Here our guess is $\thresholds=(2^{-1},2^{-2},2^{-4},2^{-4},2^{-10})$, which induces its expansion $\rexp$ and the corresponding round-up function $h$.
    }
    \label{figure:p:h}
\end{figure}

\subparagraph*{The dual program.}
The number of variables in \ref{lp:normnorm:primal} is exponential but the number of constraints is polynomial, thus we consider its dual as follows, and try to solve it using the ellipsoid method.
The dual variables are $\{r_k\}_{k\in\pos},\{y_i\}_{i\in\machines},\{z_j\}_{j\in\jobs},s,t$.
\begin{align*}
	\text{max}\; & -\sum_{k\in\pos}\left(\topl{k}{\rexp}-k\rho_k\right)r_k -\sum_{i\in\machines}y_i+
	\lambda\sum_{j\in\jobs}z_j-nt
	\tag{$\dual(\thresholds,\lambda,\tau)$}\label{lp:normnorm:dual}\\
	\text{s.t.}\; & \sum_{j\in J}z_j-y_i-\sum_{k\in\pos}\left(h(\innorm_i(J)/\tau)-\rho_k\right)^+r_k\leq
	s\cdot\mathbf{1}[\innorm_i(J)>\tau]+t \quad\forall i\in\machines,J\subseteq\jobs\\
	& r,s,t,y,z \geq0.
\end{align*}	

It is unclear how to separate the dual constraints exactly. We transform \ref{lp:normnorm:dual} into another feasibility problem that is easier to deal with.
We observe the following: if \ref{lp:normnorm:primal} is feasible, the optimum of \ref{lp:normnorm:dual} is also zero. 
Because the dual program has a trivial zero solution and the constraints are scale-invariant (i.e., if $(r,s,t,y,z)$ is feasible, $(cr,cs,ct,cy,cz)$ is feasible for any $c\geq0$), it either has optimum zero or is unbounded. 
Thus \ref{lp:normnorm:primal} is feasible if and only if \ref{lp:normnorm:dual} is bounded, which is equivalent to the following polytope being empty.
\begin{align}
	&\Big\{(r,s,t,y,z)\geq0\,\big|\tag{$\anotherdual(\thresholds,\lambda,\tau)$}\label{lp:q}\\
	&-\sum_{k\in\pos}\left(\topl{k}{\rexp}-k\rho_k\right)r_k-\sum_{i\in\machines}y_i+
	\lambda\sum_{j\in\jobs}z_j-nt\geq1;\notag\\
	&-y_i\leq s\cdot\mathbf{1}[\innorm_i(J)>\tau]+t+\sum_{k\in\pos}
	\left(h(\innorm_i(J)/\tau)-\rho_k\right)^+r_k-\sum_{j\in J}z_j,\,\forall i\in\machines,J\subseteq\jobs\Big\}.\notag
\end{align}

\begin{observation}\label{observation:q}
	\ref{lp:normnorm:primal} is feasible if and only if \ref{lp:q} is empty.
\end{observation}

\subsection{Rounding and Analysis}\label{section:main:rounding}

Next, we employ a round-or-cut argument, similar to \cite{anegg2020technique}. Intuitively speaking, we run the ellipsoid algorithm on a particular \ref{lp:q}. 
In each iteration, we either find a separating hyperplane, or directly certify that another $\anotherdual(\thresholds,\lambda',\tau')$ is non-empty, which implies that $\primal(\thresholds,\lambda',\tau')$ is infeasible by \cref{observation:q}. 
We then pick $\thresholds$ properly so that the latter cannot happen, therefore certify in polynomial time that \ref{lp:q} is indeed empty. This in turn helps us to efficiently compute a feasible solution to \ref{lp:normnorm:primal}.

For clarity of presentation, we focus on two sets of parameters in the remainder of this section, $(\lambda,\tau)\in\{(1/2,3/2),(1,1)\}$. We emphasize that this only affects the constant factors in the final analysis. We need the following core lemma on an approximate separation oracle for $\anotherdual(\thresholds,1/2,3/2)$, and defer its proof to \cref{section:oracle}. The rest of this section is dedicated to the proof of our main theorem.

\begin{lemma}\label{lemma:approx:separation}
Fix $\thresholds$. There exists a polynomial time algorithm that, given $(r,s,t,y,z)\geq0$, either outputs a violated constraint in $\anotherdual(\thresholds,1/2,3/2)$, or certifies that $\anotherdual(\thresholds,1,1)$ is non-empty.
\end{lemma}

\begin{proof}[Proof of \cref{theorem:main:normnorm}]
We enumerate all possible estimates $\thresholds$. 
Consider the choice of $\thresholds=\{\rho_k\}_{k\in\pos}$ such that the following holds,
\begin{enumerate}[label=(\roman*)]
	\item $\rho_k\in[\vco^\da_k,2\vco^\da_k)$ for $k\in\pos$ s.t. $\vco^\da_k\geq\vco^\da_1/(2m)$,
	\item $\rho_k=2^{\ceil{\log_2(\vco^\da_1/(2m))}}\in[\vco^\da_1/(2m),\vco^\da_1/m]$  for $k\in\pos$ s.t. $\vco^\da_k<\vco^\da_1/(2m)$.
\end{enumerate} 

Notice that such a vector is uniquely determined by $\vco$. We denote it by $\thresholds^\star=\{\rho_k^\star\}_{k\in\pos}$. 
We show that our exhaustive search is guaranteed to run procedures on $\thresholds^\star$. 
First, according to our initial guesses, one has $\vco^\da_1\geq \innorm_{i^\star}(\{j^\star\})=1/n$, $\vco^\da_1\leq1$, therefore (recall that each number in $\thresholds^\star$ is an integer power of 2) $\rho_1^\star\leq1$, $\rho_k^\star\geq\vco^\da_1/(2m)\geq1/(2mn)$ for each $k\in\pos$, all falling into our guessing range $[1/(2mn),1]$. Then, there are at most $n$ non-zero entries in $\vco$, since the number of jobs is $n$. Hence, for $k\in\pos$ and $k>n$, one has $\vco^\da_k=0$, and the number of distinct entries indexed by $\pos\cap[n]$ is at most $\log_2 n+3$. This implies that the number of distinct entries in $\thresholds^\star$ is at most $\log_2 n+4$. Suppose $\thresholds=\thresholds^\star$ in what follows.

For a starting point $(r,s,t,y,z)$, we repeatedly call the separation oracle in \cref{lemma:approx:separation} and use the ellipsoid algorithm to modify the solution. 
If the oracle always returns a violated constraint as a separating hyperplane, we eventually obtain that, for a polynomial-sized subset $\mcal{H}\subseteq\machines\times 2^{\jobs}$, the following polytope is empty,
\begin{align}
	&\Big\{(r,s,t,y,z)\geq0\,\big|\tag{$\anotherdual_{\mcal{H}}(\thresholds,1/2,3/2)$}\label{lp:qh}\\
	&-\sum_{k\in\pos}\left(\topl{k}{\rexp}-k\rho_k\right)r_k-\sum_{i\in\machines}y_i+\frac{1}{2}\sum_{j\in\jobs}z_j-nt\geq1;\notag\\
	&-y_i\leq s\cdot\mathbf{1}[\innorm_i(J)>3/2]+t+\sum_{k\in\pos}
	\left(h(2\innorm_i(J)/3)-\rho_k\right)^+r_k-\sum_{j\in J}z_j,\,\forall (i,J)\in\mcal{H}\Big\}.\notag
\end{align}
Following a similar argument as \cref{observation:q}, by removing constraints in $\dual(\thresholds,1/2,3/2)$ that are indexed by $\machines\times2^\jobs\setminus\mcal{H}$ and obtaining a polynomial-sized LP $\dual_{\mcal{H}}(\thresholds,1/2,3/2)$, the latter is bounded. Thus, its dual $\primal_{\mcal{H}}(\thresholds,1/2,3/2)$ is feasible with a polynomial number of variables indexed by $\mcal{H}$. 
We directly solve $\primal_{\mcal{H}}(\thresholds,1/2,3/2)$ and obtain a \emph{vertex solution} $\hat x$.
It is easy to see that $\hat x$ is feasible to $\primal(\thresholds,1/2,3/2)$, because $\primal_{\mcal{H}}(\thresholds,1/2,3/2)$ can be obtained from $\primal(\thresholds,1/2,3/2)$ by eliminating all variables that are NOT indexed by $\mcal{H}$.

We look at the tight constraints at $\hat x$ among
\eqref{lp:normnorm:primal2} and \eqref{lp:normnorm:primal5}.
If $m\leq n$, the number is at most $m+1$; if $m>n$, the number of tight ones in \eqref{lp:normnorm:primal2} is at most $n$, otherwise \eqref{lp:normnorm:primal5} is violated.
Hence there are at most $\min\{m,n\}+1$ tight constraints among them at $\hat x$. 
Since there are at most $\log_2n+4$ distinct entries in $\thresholds$ hence \eqref{lp:normnorm:primal1}, the total number of non-trivial tight constraints at $\hat x$ is at most $3n+5$, 
and the support size of $\hat x$ is at most $3n+5$ (see, e.g., \cite{lau2011iterative}). Denote the support of $\hat x$ by $\mcal{\hat H}\subseteq\mcal{H}$.

Another possibility is that, in some iteration of the ellipsoid algorithm, the oracle certifies that $\anotherdual(\thresholds,1,1)$ is non-empty, and $\primal(\thresholds,1,1)$ is infeasible due to \cref{observation:q}. 
We show that this is impossible given the choice of $\thresholds=\thresholds^\star$. 
Indeed, we first notice the following.
\begin{enumerate}[label=(\roman*)]
	\item for $k\in\pos$, since $\rho_k\in \thresholds$ and $\rho_k\geq\vco^\da_k$, we have $\rexpsub{k}=\rho_k\geq h(\vco^\da_k)$,
	\item $\rexpsub{k}=\rho_{\prevp{k}}\geq h(\vco^\da_{\prevp{k}})\geq h(\vco^\da_k)$ when $k\notin\pos$,
\end{enumerate} 
thus $\rexp$ and $h(\vco^\da)$ (apply $h$ element-wise) are both non-increasing and $\rexp\geq h(\vco^\da)$.
From the optimal assignment $\sigma^\star:\jobs\rightarrow\machines$, we let $x_{i,J}^\star=1$ if $J=\{j:\sigma^\star(j)=i\}$ and 0 otherwise. 
Fix $k\in\pos$ in \eqref{lp:normnorm:primal1} of $\primal(\thresholds,1,1)$. 
The LHS at $x^\star$ is
\(
k\rho_k+\sum_{s\in[m]}(h(\vco^\da_s)-\rho_k)^+
\leq k\rho_k+\sum_{s\in[m]}\left(\rexpsub{s}-\rho_k\right)^+
=\topl{k}{\rexp},
\)
via \cref{claim:topl}, the definition of $\rexp$, and that $\rexp\geq h(\vco^\da)$ is non-increasing.
Other constraints in $\primal(\thresholds,1,1)$ are easily satisfied by $x^\star$.
Therefore $\primal(\thresholds,1,1)$ must be feasible, a contradiction.

In conclusion, when $\thresholds=\thresholds^\star$, we obtain a feasible solution $\hat x$ to $\primal(\thresholds,1/2,3/2)$ with support size $|\mcal{\hat H}|\leq 3n+5$. We proceed to use randomized rounding to obtain a feasible assignment of jobs.
Let $\mcal{I}\leftarrow\emptyset$ and $T=\ceil{6\ln n}$. For each $t=1,\dots,T$ and each $(i,J)\in \mcal{\hat H}$, set $\mcal{I}\leftarrow \mcal{I}\cup\{(i,J)\}$ independently with probability $\hat x_{i,J}$. 

\subparagraph{The Machine Duplicates.} For each $i\in\machines$, we have $\sum_{J:(i,J)\in\mcal{\hat H}}\hat x_{i,J}\leq 1$ using the constraint \eqref{lp:normnorm:primal2}.
Therefore for the number of appearances of $i$ in $\mcal{I}$, one has
\begin{equation*}
	\E[|\{J:(i,J)\in \mcal{I}\}|]=\sum_{J:(i,J)\in\mcal{\hat H}}1-(1-\hat x_{i,J})^{T}\leq T\sum_{J:(i,J)\in\mcal{\hat H}}\hat x_{i,J}\leq T,
\end{equation*}
using Bernoulli's inequality. Hence, using Chernoff bound in \cref{lemma:chernoff}, on random variables $\mathbf{1}[(i,J)\in\mcal{I}],\,(i,J)\in\mcal{\hat H}$, one has
\begin{equation}
\Pr[|\{J:(i,J)\in \mcal{I}\}|> 6T]
=\Pr\left[\sum_{J:(i,J)\in\mcal{\hat H}}\mathbf{1}[(i,J)\in\mcal{I}]> 6T\right]
\leq 2^{-6T}\leq1/(n^{24}).\label{eqn:repeat:machine}
\end{equation}
Note that since $|\mcal{\hat H}|\leq 3n+5$, at most $3n+5$ machines can possibly have a non-zero number of appearances in $\mcal{I}$.
For any other machine $i$, we always have $|\{J:(i,J)\in \mcal{I}\}|=0$. 

\subparagraph{The Coverage of Jobs.} For each $j\in\jobs$, the probability of it \emph{not} being contained in any selected $J$ in $\mcal{I}$ is
\begin{align}
\prod_{(i,J)\in\mcal{\hat H}:J\ni j}(1-\hat x_{i,J})^T & \leq\prod_{(i,J)\in\mcal{\hat H}:J\ni j}\exp(-T\hat x_{i,J})\notag\\
&=\exp\Big(-T\sum_{(i,J)\in\mcal{\hat H}:J\ni j}\hat x_{i,J}\Big)\leq\exp(-T/2)\leq1/n^3,\label{eqn:cover:job}
\end{align}
where we use \eqref{lp:normnorm:primal3} and obtain $\sum_{(i,J)\in\mcal{\hat H}:J\ni j}\hat x_{i,J}\geq1/2$ in the penultimate inequality.

\subparagraph{The Norm Constraints.} According to \eqref{lp:normnorm:primal4}, $\hat x_{i,J}>0$ implies $\innorm_i(J)\leq3/2$ thus $h(2\innorm_i(J)/3)\leq h(1)=1$.
Define subsets $\mcal{\hat H}_t=\{(i,J)\in\mcal{\hat H}:h(2\innorm_i(J)/3)=2^t\}$ where $t=0$ or $2^t$ is an entry in $\thresholds$. 
Let the subsets be indexed by $t\in\pwr\subseteq\mathbb{Z}_{\leq0}$, then $|\pwr|\leq\log_2 n+5$ since there are at most $\log_2 n+4$ distinct entries in $\thresholds$, and $\bigcup_{t\in\pwr}\mcal{\hat H}_t=\mcal{\hat H}$ because $h(2\innorm_i(J)/3)\leq1$ for each $(i,J)\in\mcal{\hat H}$. 
For each $t\in\pwr$, we define a random variable $Y_t=|\mcal{\hat H}_t\cap \mcal{I}|$ and obtain
\[
\E[Y_t]=\sum_{(i,J)\in\mcal{\hat H}_t}1-(1-\hat x_{i,J})^T\leq T\sum_{(i,J)\in\mcal{\hat H}_t}\hat x_{i,J},
\] 
via Bernoulli's inequality. Using Markov's inequality, with probability at least 1/2, one has $Y_t\leq 2\E[Y_t]$. Since these (at most) $\log_2 n+5$ random variables $\{Y_t:t\in\pwr\}$ are independent, there is a probability at least $1/(32 n)$ that $Y_t\leq 2\E[Y_t]$ for all $t\in\pwr$. Suppose this happens. Then for each $k\in\pos$, because $\hat x$ satisfies \eqref{lp:normnorm:primal1}, we have
\begin{align}
	&\sum_{(i,J)\in \mcal{I}}(h(2\innorm_i(J)/3)-\rho_k)^+
	=\sum_{t\in\pwr}\sum_{(i,J)\in \mcal{I}\cap\mcal{\hat H}_t}(2^t-\rho_k)^+=\sum_{t\in\pwr}Y_t\cdot (2^t-\rho_k)^+\notag\\
	&\leq \sum_{t\in\pwr}2\E[Y_t](2^t-\rho_k)^+\leq 2T\sum_{t\in\pwr}\sum_{(i,J)\in\mcal{\hat H}_t}\hat x_{i,J}(2^t-\rho_k)^+\leq 2T\left(\topl{k}{\rexp}-k\rho_k\right).\label{eqn:norm:constraints}
\end{align}

\subparagraph{Putting it all together.}
Using \eqref{eqn:repeat:machine}\eqref{eqn:cover:job}\eqref{eqn:norm:constraints} and the union bound, for \emph{large enough} $n$, with probability at least $1/(32n)-(3n+5)/(n^{24})-1/n^2\geq 1/(64n)$, the following facts hold:  
\begin{enumerate}[label=(\alph*)]
	\item\label{item:machine:duplicate} for each $i\in\machines$, $|\{J:(i,J)\in \mcal{I}\}|\leq 6T\leq 38\ln n$,
	\item\label{item:job:coverage} for each $j\in\jobs$, $\exists (i,J)\in \mcal{I}$ s.t. $j\in J$,
	\item\label{item:norm:constraints} for each $k\in\pos$, $\sum_{(i,J)\in \mcal{I}}\left(h(2\innorm_i(J)/3)-\rho_k\right)^+\leq 14\ln n\left(\topl{k}{\rexp}-k\rho_k\right)$.
\end{enumerate}

We repeat the randomized rounding $64n$ times, and boost the success probability to at least $1-(1-1/(64n))^{64n}\geq1-e^{-1}\geq0.6$.
Suppose \ref{item:machine:duplicate}\ref{item:job:coverage}\ref{item:norm:constraints} hold in the rest of this section.
Next, we merge all configurations that are identified with the same machine in $\mcal{I}$, that is, $J_i\leftarrow\bigcup_{J:(i,J)\in \mcal{I}}J$ and $\bigcup_{i\in\machines}J_i=\jobs$ according to \ref{item:job:coverage}. 
We only merge $O(\log n)$ configurations for each machine using \ref{item:machine:duplicate}. 
We also have the following subadditivity 
\[\innorm_i(J_1)+\innorm_i(J_2)
\geq \innorm_i(\boldp_i\angles{J_1}+\boldp_i\angles{J_2})
\geq \innorm_i(\boldp_i\angles{J_1\cup J_2})
=\innorm_i(J_1\cup J_2),
\]
from the triangle inequality on the norm $\innorm_i$, and $s^++t^+\geq(s+t)^+$ for any $s,t\in\mathbb{R}$. Since $h(x)\geq x$ for all $x\geq0$, one has for each $k\in\pos$,
\begin{align*}
	&\sum_{i\in\machines}\left(2\innorm_i(J_i)/3-38\ln n\cdot\rho_k\right)^+\leq_{\ref{item:machine:duplicate}}\sum_{i\in\machines}\left(\sum_{J:(i,J)\in\mcal{I}}2\innorm_i(J)/3-|\{J:(i,J)\in\mcal{I}\}|\cdot\rho_k
	\right)^+\\
	&\leq\sum_{i\in\machines}\sum_{J:(i,J)\in\mcal{I}}\left(h(2\innorm_i(J)/3)-\rho_k\right)^+\leq_{\ref{item:norm:constraints}} 14\ln n\left(\topl{k}{\rexp}-k\rho_k\right).
\end{align*}

Define a vector $\vcv\in\Rpos^m$ where $\vcv_i=\innorm_i(J_i)$ for each $i\in\machines$. For each $k\in\pos$, via \cref{claim:topl}, we have
\begin{align*}
\topl{k}{\vcv}&\leq k\cdot (57\ln n\cdot\rho_k)+\sum_{i\in\machines}(\innorm_i(J_i)-57\ln n\cdot\rho_k)^+\\
&\leq 57\ln n\cdot k\rho_k+21\ln n\left(\topl{k}{\rexp}-k\rho_k\right)\leq 57\ln n\cdot\topl{k}{\rexp}.
\end{align*}

Then for $k\notin\pos$, one has $k<\nextp{k}<2k$ and thus
\[
\topl{k}{\vcv}\leq\topl{\nextp{k}}{\vcv}
\leq 57\ln n\cdot\topl{\nextp{k}}{\rexp}\leq 114\ln n\cdot\topl{k}{\rexp},
\]
and the two inequalities above show that $\outnorm(\vcv)\leq O(\log n)\cdot\outnorm(\rexp)$ using \cref{lemma:majorization}. By our initial assumption that $\thresholds=\thresholds^\star$, we have $\rexpsub{k}\leq2\vco^\da_k+\vco^\da_1/m$ for each $k\in\pos$, and we now compare $\outnorm(\rexp)$ and $\opt=\outnorm(\vco)$. First for $k\in\pos$, because $\rexp$ is non-increasing, we have
\begin{align*}
\topl{k}{\rexp}&=\sum_{s<k,s\in\pos}\rexpsub{s}+\sum_{s<k,s\notin \pos}\rexpsub{s}+\rexpsub{k}\\
&\leq2\sum_{s<k,s\in\pos}\vco^\da_{s}+2\sum_{s<k,s\notin \pos}\vco^\da_{\prevp{s}}+2\vco^\da_k+k\cdot\frac{\vco^\da_1}{m}\\
&\leq\vco^\da_1+2\vco^\da_k+2\sum_{s<k,s\in\pos}(\nextp{s}-s)\vco^\da_s\\
&\leq\vco^\da_1+2\vco^\da_k+2\vco^\da_1+2\sum_{1<s<k,s\in\pos}2(s-\prevp{s})\vco^\da_s\\
&\leq4\sum_{s\leq k,s\in\pos}(s-\prevp{s})\vco^\da_s\leq4\sum_{s'\leq k}\vco^\da_{s'}=4\topl{k}{\vco},
\end{align*}
where we use that $\vco^\da$ is non-increasing, and $\nextp{k}-k\leq 2(k-\prevp{k})$ for each $k\in\pos$, by definition of $\pos$. Then for $k\notin\pos$, likewise we obtain
\[
\topl{k}{\rexp}\leq\topl{\nextp{k}}{\rexp}\leq4\topl{\nextp{k}}{\vco}\leq8\topl{k}{\vco},
\]
whence it follows that $\outnorm(\rexp)\leq8\outnorm(\vco)$ by \cref{lemma:majorization} again, and $\outnorm(\vcv)\leq O(\log n)\cdot\outnorm(\vco)$. Since the norms are monotone and $\bigcup_{i\in\machines}J_i=\jobs$, any assignment induced by these subsets reveals the same approximation ratio. That is, one can easily obtain a job assignment $\sigma:\jobs\rightarrow\machines$ such that $\sigma^{-1}(i)\subseteq J_i$ for each $i\in\machines$, thus $\boldp_i\assigned{\sigma}\leq\boldp_i\angles{J_i}$ and
\[
\makespan(\sigma)
=\outnorm(\{\innorm_i(\boldp_i\assigned{\sigma})\}_{i\in\machines})
\leq \outnorm(\{\innorm_i(\boldp_i\angles{J_i})\}_{i\in\machines})
=\outnorm(\vcv)\leq O(\log n)\cdot \opt.\qedhere
\]
\end{proof}

\section{The Approximate Separation Oracle}\label{section:oracle}

In this section, we provide details of the claimed separation oracle in \cref{lemma:approx:separation}. As a useful subroutine, we consider the following norm minimization with a linear constraint problem and obtain a PTAS, which may be of independent interest.

\begin{definition}
	(\normval). Given a symmetric monotone norm $\innorm:\R^n\rightarrow\Rpos$, two non-negative $n$-dimensional vectors $\boldp=\{p_j\}_{j\in[n]},\,\boldz=\{z_j\}_{j\in[n]}$, and a non-negative real number $Z$, the goal is to find $J\subseteq[n]$ such that $\sum_{j\in J}z_j\geq Z$, so as to minimize $\innorm(\boldp\angles{J})$.\end{definition}

\begin{theorem}\label{theorem:norm:value}
There exists a polynomial time approximation scheme for \emph{\normval}.
\end{theorem}

We first show the separation oracle in \cref{section:proof:separation:oracle}, then prove \cref{theorem:norm:value} in \cref{section:normval}.

\subsection{Proof of \cref{lemma:approx:separation}}
\label{section:proof:separation:oracle}

Fix the solution $(r,s,t,y,z)$. If it violates the first constraint in $\anotherdual(\thresholds,1/2,3/2)$, we directly choose it as a separating hyperplane. If not, we try to approximately verify the second set of constraints, which is the crux of this lemma. 

We notice that, for each such constraint, the LHS is determined by $i$, and the RHS is determined by $\innorm_i(J)$ and $\sum_{j\in J}z_j$. Therefore, we consider a surrogate optimization problem as follows: for each $Z\geq0$, we want to check whether there exist $i\in\machines$, $J\subseteq\jobs$ such that $\sum_{j\in J}z_j\geq Z$, $\innorm_i(J)$ is minimized, and the corresponding constraint in $\anotherdual(\thresholds,1/2,3/2)$ is violated. Unfortunately, there are an exponential number of possible values for $Z$, so we can only consider a small subset of approximate values.

Fix $i\in\machines$. We enumerate the largest variable $z_j$ in the chosen set $j\in J$, and consider minimizing $\innorm_i(J)$ restricted to jobs $\jobs_{\leq j}:=\{j'\in\jobs:z_{j'}\leq z_j\}$. 
After we fix $z_j$, the possible values of $\sum_{j\in J}z_j$ fall in the interval $[z_j,nz_j]$, so we enumerate $Z\in\{z_j,2z_j,4z_j,\dots,2^s z_j\}$, where $s\in\mathbb{Z}$ is the smallest integer such that $2^s\geq n$. The problem now becomes to minimize $\innorm_i(J)$ subject to $\sum_{j\in J}z_j\geq Z$ and $J\subseteq \jobs_{\leq j}$.

This is an instance of \normval. By \cref{theorem:norm:value}, for each sub-problem above, if it is feasible w.r.t. $Z$, in polynomial time we obtain $J$ such that $\sum_{j\in J}z_j\geq Z$ and $\innorm_i(J)$ is at most $3/2$ times the optimum.
If, among the output solutions, there exist $i,J$ that violate the said constraint in $\anotherdual(\thresholds,1/2,3/2)$, we output the corresponding separating hyperplane; otherwise, though the optimization above is approximate, we still certify the following,
\begin{equation}
	-y_i\leq s\cdot\mathbf{1}[\innorm_i(J)>1]+t+\sum_{k\in\pos}\left(h(\innorm_i(J))-\rho_k\right)^+r_k-\frac{1}{2}\sum_{j\in J}z_j,\,\forall i\in\machines,J\subseteq\jobs.
	\tag{*}\label{eqn:certified}
\end{equation}

To see this, for the sake of contradiction, suppose there exist $i,J$ that violate \eqref{eqn:certified}. When we enumerate the correct $j=\argmax_{j'\in J}\{z_{j'}\}$ and $Z$ such that $\sum_{j\in J}z_j\in[Z,2Z)$, since $J$ is a feasible solution in this case, according to our supposition, we obtain $J'$ such that $\sum_{j\in J'}z_j\geq Z$ and $\innorm_i(J')\leq 3\innorm_i(J)/2$. But using the assumption of violating \eqref{eqn:certified}, one has
\begin{align*}
	& s\cdot\mathbf{1}[\innorm_i(J')>3/2]+t+
	\sum_{k\in\pos}\left(h\left(\frac{\innorm_i(J')}{3/2}\right)-\rho_k\right)^+r_k-\sum_{j\in J'}z_j\\
	&\leq s\cdot\mathbf{1}[\innorm_i(J)>1]+t+\sum_{k\in\pos}\left(h(\innorm_i(J))-\rho_k\right)^+r_k-Z\\
	&\leq s\cdot\mathbf{1}[\innorm_i(J)>1]+t+\sum_{k\in\pos}\left(h(\innorm_i(J))-\rho_k\right)^+r_k-\frac{1}{2}\sum_{j\in J}z_j<-y_i,
\end{align*}
violating this constraint in $\anotherdual(\thresholds,1/2,3/2)$, and we should output this hyperplane defined by $(i,J')$ in the first place, which is a contradiction.

Let $(r',s',t',y',z')=(r,s,t,y,z/2)$. Since by now the first constraint in $\anotherdual(\thresholds,1/2,3/2)$ must be satisfied by $(r,s,t,y,z)$, using \eqref{eqn:certified}, it is easy to verify that $(r',s',t',y',z')$ satisfies all the constraints in $\anotherdual(\thresholds,1,1)$, hence it is non-empty.\qed


\subsection{Norm Minimization with a Linear Constraint}
\label{section:normval}

W.l.o.g., assume $\boldp>0$, otherwise the problem can be easily reduced to such an instance. 
We fix a small constant $\epsilon\in(0,1/2]$ and round each $p_j$ \emph{up} to its nearest integer power of $1+\epsilon$ and solve the new instance. This causes us to lose a factor of at most $1+\epsilon$ in the objective since $\innorm$ is a monotone norm.
Fix an unknown optimal solution $J^\star$ to the \emph{modified instance} with $\vco=\boldp\angles{J^\star}$ and $\innorm(\vco)\leq(1+\epsilon)\opt$, where $\opt\geq0$ is the optimum in the \emph{original instance}. 

To obtain a solution $J$ with objective $\innorm(\boldp\angles{J})$ bounded by $(1+O(\epsilon))\innorm(\vco)$, we use the same technique as in \cref{section:normnorm} and consider the $\topp{k}$ norms of $\boldp\angles{J}$.
As before, we start with some guessing procedures similar to \cref{section:main:lp}, but with several subtle differences for obtaining a PTAS. 
Note that we abuse some of the notations in \cref{section:main:lp}
for convenience. They play very similar roles, but may be defined slightly differently from those in \cref{section:normnorm}.

\subparagraph*{Guessing the optimum.}
We guess the largest $p_{j_1},j_1\in J^\star$, the largest $z_{j_2},j_2\in J^\star$, and suppose w.l.o.g. that $\vco^\da_1=p_{j_1}=1$, $z_{j_2}=1$. 
We assume $j_1,j_2$ are correctly guessed in the sequel (it is easy to see there are only polynomial number of possible choices). 
	
Define indexes $\pos\subseteq[n]$ iteratively as follows \cite{ibrahimpur2021minimum}:
Set $\pos\leftarrow\{1\}$. Whenever $n\notin\pos$, choose the current largest $t\in\pos$ and add $\ceil{(1+\epsilon)t}$ to $\pos$ but no larger than $n$, i.e., \[\pos\leftarrow\pos\cup\{\min\{n,\ceil{(1+\epsilon)\max\{t\in\pos\}}\}\},\] 
and it is easy to verify $|\pos|=O(\log_{1+\epsilon}n)$. 
For each $k\in[n]$, define $\nextp{k}$ as the smallest number in $\pos$ that is larger than $k$, and $\nextp{n}=n+1$; $\prevp{k}$ as the largest number in $\pos$ that is smaller than $k$, and $\prevp{1}=0$. Using $\pos$, we guess the following. 
	\begin{enumerate}[label=(\arabic*)]
		\item A non-increasing vector $\thresholds=\{\rho_k\}_{k\in\pos}\in\Rpos^\pos$ such that each entry is a \emph{non-positive} integer power of $(1+\epsilon)$, and
	\begin{enumerate}[label=(\roman*)]
		\item $\rho_k=\vco^\da_k$ for $k\in\pos$ s.t. $\vco^\da_k\geq\epsilon/n$,
		\item $\rho_k=(1+\epsilon)^{\ceil{\log_{1+\epsilon}(\epsilon/n)}}$ for $k\in\pos$ s.t. $\vco^\da_k<\epsilon/n$.
	\end{enumerate} 
	Since $|\pos|=O(\log_{1+\epsilon} n)$ and the number of possible values is $O(\log_{1+\epsilon}(n/\epsilon))$, the number of such non-increasing vectors is $\binom{O(\log_{1+\epsilon}(n/\epsilon))}{O(\log_{1+\epsilon}n)}\leq 2^{O(\log_{1+\epsilon} (n/\epsilon))}=(n/\epsilon)^{O(1/\epsilon)}$, using a basic counting method as before. 
	We use exhaustive search and assume $\thresholds$ is correct w.r.t. $\vco^\da$ in what follows. Define a non-decreasing function $h:\Rpos\rightarrow\Rpos$ where
	\begin{equation*}
	h(x)=\left\{\begin{array}{cl}
		\min\{t\in \thresholds:t\geq x\}	 & x\leq\max\{t\in \thresholds\}\\
		x & x>\max\{t\in \thresholds\}.
	\end{array}
	\right.
	\end{equation*}
	We notice that $h(x)\geq x$ holds for any $x\geq0$, and $h(x)\leq1$ for each $x\leq1$ since $\rho_1=\vco^\da_1=1$ is the largest entry in $\thresholds$ according to our guesses.
	\item We start with $\pwr\leftarrow\emptyset$. Given $\thresholds$ and $h$, for each $s\in\mathbb{Z}$ such that $(1+\epsilon)^s\in \thresholds$, $\pwr\leftarrow\pwr\cup\{s\}$. We have $\pwr\subseteq\mathbb{Z}_{\leq0}$ since $\vco^\da_1=1$ and thus $\thresholds\subseteq[0,1]^\pos$. 
	
	Define a partition $\mcal{C}=\{C_t:t=-1,0,\dots,\ceil{1/\epsilon}\}$ of $\pwr$ as follows: For each $s\in\pwr$, guess whether the \emph{number} of indexes in $J^\star$ (i.e., the optimum) such that $h(p_j)=(1+\epsilon)^s$ is $>\ceil{1/\epsilon}$. 
	If so, $C_{-1}\leftarrow C_{-1}\cup\{s\}$; otherwise, guess $0\leq t\leq\ceil{1/\epsilon}$ as this number and $C_t\leftarrow C_t\cup\{s\}$. Since $|\pwr|\leq|\pos|=O(\log_{1+\epsilon} n)$ and there are $\ceil{1/\epsilon}+2$ classes, the number of possible partitions is at most $(1/\epsilon)^{O(\log_{1+\epsilon}n)}=n^{O(\epsilon^{-1}\log(1/\epsilon))}$. 
	We use exhaustive search and assume $\mcal{C}$ is correct w.r.t. $J^\star$ in what follows.
	\end{enumerate}
	
We use $\thresholds$ to construct an entry-wise upper bound for $\vco$, using which we write linear constraints and bound the $\topp{k}$ norms of an LP solution, $k\in\pos$;
$h$ is a round-up function that naturally classifies $[n]$ into $O(\log_{1+\epsilon}n)$ classes according to $h(p_j),j\in[n]$.
Crucially, we use $\mcal{C}$ to guess the number of \emph{selected indexes} of each class in the optimum $J^\star$, and only attempt to exactly match the numbers for those classes with a \emph{selected cardinality} $\leq\ceil{1/\epsilon}$. 
That is, if a class contains $>\ceil{1/\epsilon}$ indexes in $J^\star$, it is added to $C_{-1}$. This ``meta-classification'' of classes is useful in our analysis of the approximation factor for \emph{deterministic} rounding.

Similar as before, we define a non-increasing expansion vector $\rexp\in\Rpos^n$, where
$\rexpsub{k}=\rho_k$ if $k\in\pos$, and $\rexpsub{k}=\rho_{\prevp{k}}$ if $k\notin\pos$. 
Since $\rho_k\geq\vco^\da_k$ for each $k\in\pos$ and $\rho_k\in \thresholds$, it follows that $\rexpsub{k}=\rho_k\geq h(\vco^\da_k)$ when $k\in\pos$, $\rexpsub{k}=\rho_{\prevp{k}}\geq h(\vco^\da_{\prevp{k}}) \geq h(\vco^\da_k)$ for each $k\notin\pos$ and thus $\rexp\geq h(\vco^\da)$. We need the following result. The lemma is implied by Lemma~2.7 in \cite{ibrahimpur2021minimum}, and we provide the proof in \cref{section:lemma:threshold:proof} for completeness. 
\begin{lemma}\label{lemma:threshold} 
\emph{(Lemma~2.7, \cite{ibrahimpur2021minimum}).}
	$\innorm(\rexp)\leq(1+11\epsilon)\innorm(\vco)$.
\end{lemma}

\subparagraph*{LP relaxation.}
The following relaxation uses the variable $x_{j}\geq0$ to represent the extent we include $j\in[n]$ in the solution. 
Similar to our main relaxation \ref{lp:normnorm:primal}, it attempts to bound the norm objective via \eqref{lp:norm:value1}, using $\topp{k}$ norms and \cref{lemma:majorization}. 
We also have constraints \eqref{lp:norm:value3}, \eqref{lp:norm:value4} that restrict the number of indexes selected in the classes $\{j\in[n]:h(p_j)=(1+\epsilon)^s\},\,s\in\pwr$, as mentioned above.
\begin{align}
	\text{min} && 0\tag{${\operatorname{NLin}}$}\label{lp:norm:value}\\
    \text{s.t.} && \sum_{j\in[n]}(h(p_{j})-\rho_k)^+x_{j}
    &\leq \topl{k}{\rexp}-k\rho_k && \forall k\in\pos
    \tag{$\operatorname{NLin.1}$}\label{lp:norm:value1}\\
    && \sum_{j\in[n]}z_jx_j &\geq Z
    \tag{$\operatorname{NLin.2}$}\label{lp:norm:value2}\\
    && \sum_{j:h(p_j)=(1+\epsilon)^s}x_j &=t &&\forall s\in C_t,0\leq t\leq\ceil{1/\epsilon}
    \tag{$\operatorname{NLin.3}$}\label{lp:norm:value3}\\
    && \sum_{j:h(p_j)=(1+\epsilon)^s}x_j &\geq1+\ceil{1/\epsilon} &&\forall s\in C_{-1}\tag{$\operatorname{NLin.4}$}\label{lp:norm:value4}\\
     && x_{j}&=0 && \forall j\text{ s.t. }p_j>1\text{ or }z_j>1
     \tag{$\operatorname{NLin.5}$}\label{lp:norm:value5}\\
     && x&\in[0,1]^{[n]}.\notag
\end{align}

\begin{claim}\label{lemma:lp:norm:value:feasible}
\ref{lp:norm:value} is feasible.
\end{claim}
\begin{proof}
Set $x_{j}^\star=1$ if $j\in J^\star$ and zero otherwise.
It is easy to see that all but the first set of constraints are satisfied by $x^\star$.
For \eqref{lp:norm:value1}, we fix $k\in\pos$. The LHS is
\[
\sum_{s\in[n]}(h(\vco^\da_s)-\rho_k)^+
\leq\sum_{s\in[n]}(\rexpsub{s}-\rho_k)^+
=\sum_{s\leq k}\left(\rexpsub{s}-\rho_k\right)^+
=\topl{k}{\rexp}-k\rho_k,
\]
where we use $h(\vco^\da)\leq\rexp$, $\rexpsub{k}=\rho_k$ and the fact that $\rexp$ is non-increasing.
\end{proof}

\begin{proof}[Proof of \cref{theorem:norm:value}]
Using \cref{lemma:lp:norm:value:feasible}, we obtain a feasible solution $\bar x$. Our rounding method is very simple. Let $J\leftarrow\emptyset$. For each $s\in\pwr$, define the partial support $J_s=\{j\in[n]:h(p_j)=(1+\epsilon)^s,\bar x_j>0\}$.
$\bigcup_{s\in\pwr}J_s$ is the entire support of $\bar x$ by \eqref{lp:norm:value5}, $\rho_1=\vco^\da_1=1$ and definition of $h$. 
We have the following cases.
\begin{enumerate}[label=(\roman*)]
	\item If $s\in C_t$ for some $0\leq t\leq\ceil{1/\epsilon}$, using \eqref{lp:norm:value3} one has $\sum_{j\in J_s}\bar x_j=t$ and thus $|J_s|\geq t$. 
	We choose $t$ indexes $J_s'\subseteq J_s$ that have the largest $z_j$ values and set $J\leftarrow J\cup J_s'$. It is easy to verify that $\sum_{j\in J_s'}z_j\geq\sum_{j\in J_s}z_j\bar x_j$.
	\item If $s\in C_{-1}$, using \eqref{lp:norm:value4} one has $\sum_{j\in J_s}\bar x_j\geq 1+\ceil{1/\epsilon}$ and $|J_s|\geq \ceil{\sum_{j\in J_s}\bar x_j}$. 
	We choose $\ceil{\sum_{j\in J_s}\bar x_j}$ indexes $J_s'\subseteq J_s$ that have the largest $z_j$ values and set $J\leftarrow J\cup J_s'$. 
	It is easy to verify that $|J_s'|/\sum_{j\in J_s}\bar x_j=\ceil{\sum_{j\in J_s}\bar x_j}/\sum_{j\in J_s}\bar x_j \leq (2+\ceil{1/\epsilon})/(1+\ceil{1/\epsilon})\leq 1+\epsilon$, and $\sum_{j\in J_s'}z_j\geq\sum_{j\in J_s} z_j\bar x_j$.
\end{enumerate}

Using \eqref{lp:norm:value2}, it immediately follows that 
\[\sum_{j\in J}z_j
=\sum_{s\in\pwr}\sum_{j\in J_s'}z_j\geq\sum_{s\in\pwr}\sum_{j\in J_s}z_j\bar x_j
=\sum_{j\in[n]}z_j\bar x_j\geq Z,
\]
showing that $J$ is a feasible solution. For the objective, we fix $k\in\pos$. One has
\begin{align*}
	\sum_{j\in J}(h(p_j)-\rho_k)^+
	&=\sum_{s\in\pwr}\sum_{j\in J_s'}(h(p_j)-\rho_k)^+=\sum_{s\in\pwr}|J_s'|((1+\epsilon)^s-\rho_k)^+\\
	&\leq\sum_{s\in\pwr}(1+\epsilon)\sum_{j\in J_s}((1+\epsilon)^s-\rho_k)^+\bar x_j\\
	&=(1+\epsilon)\sum_{j\in [n]}(h(p_j)-\rho_k)^+\bar x_j\leq(1+\epsilon)\left(\topl{k}{\rexp}-k\rho_k\right),
\end{align*}
where we use \eqref{lp:norm:value1} in the last inequality.
Using $h(x)\geq x$ and \cref{claim:topl}, for each $k\in\pos$,
\begin{align}
    \topl{k}{\boldp\angles{J}}&\leq k\rho_k+\sum_{s\in [n]}(\boldp\angles{J}_s-\rho_k)^+\leq k\rho_k+\sum_{j\in J}(h(p_{j})-\rho_k)^+\notag\\
    &\leq k\rho_k+(1+\epsilon)\left(\topl{k}{\rexp}-k\rho_k\right)\leq(1+\epsilon)\topl{k}{\rexp}.\label{eqn:major:1}
\end{align}

Then for $k\notin\pos$, by considering $\prevp{k}$ and the iterative definition of $\pos$, we have $(k-1)(1+\epsilon)>k$, thus $k>1+1/\epsilon$, and $\nextp{k}<1+(1+\epsilon)k\leq(1+2\epsilon)k$. Since $\boldp\angles{J}^\da$ and $\rexp$ are non-increasing, and $k<\nextp{k}<(1+2\epsilon)k$, using \eqref{eqn:major:1} we have
\begin{align}
    \topl{k}{\boldp\angles{J}}\leq\topl{\nextp{k}}{\boldp\angles{J}}
    \leq (1+\epsilon)\topl{\nextp{k}}{\rexp}
    \leq (1+5\epsilon)\topl{k}{\rexp},\label{eqn:major:2}
\end{align}
thus from \eqref{eqn:major:1}\eqref{eqn:major:2} and \cref{lemma:majorization}, \cref{lemma:threshold}, it follows that 
\[\innorm(\boldp\angles{J})\leq(1+5\epsilon)\innorm(\rexp)\leq(1+71\epsilon)\innorm(\vco)\leq(1+143\epsilon)\opt.
\]
Finally, it is easy to see that the running time of the algorithm is $(n/\epsilon)^{O(\epsilon^{-1}\log(1/\epsilon))}$, which is determined by our guessing procedure.
\end{proof}

\section{Conclusion and Future Directions}
\label{section:conclusion}

In this work, we systematically study the approximation algorithms for \glb with general 
inner and outer norms.
We propose a randomized polynomial time algorithm with logarithmic approximation factor, matching the lower bound up to constant. For certain special case that generalizes classic makespan minimization, we develop a constant factor approximation algorithm.

We propose some interesting future directions.
Note the $\Omega(\log n)$ lower bound in \cref{theorem:lower:bound:glb} does not necessarily hold for special cases with outer norm $\outnorm=\mcal{L}_\infty$ and arbitrary symmetric monotone inner norms.
Hence, there is a gap of $\log n$ between the known lower bound and our result (see \cref{figure:result:table}), which is an interesting open question.
Other natural inner and outer objectives are also worth studying, e.g., submodular set functions \cite{svitkina2011submodular,svitkina2010facility}, piece-wise linear functions \cite{li2011generalized}, etc.

\bibliography{references}

\appendix
\section{Proof of \cref{lemma:threshold}}
\label{section:lemma:threshold:proof}
	By our guesses, we have $\rho_k=\vco^\da_k$ if $\vco^\da_k\geq \epsilon/n$ and $\rho_k=(1+\epsilon)^{\ceil{\log_{1+\epsilon}(\epsilon/n)}}$ otherwise. For each $k\in\pos$, one then has $\rho_k\leq\vco^\da_k+2\epsilon/n$ and
\begin{align}
    \topl{k}{\rexp}
    &\leq\sum_{s\leq k,s\in\pos}\vco^\da_{s}+\sum_{s< k,s\notin\pos}\vco^\da_{\prevp{s}}+k\cdot\frac{2\epsilon}{n}\notag\\
    &\leq2\epsilon+\vco^\da_k+\sum_{s< k,s\in\pos}(\nextp{s}-s)\vco^\da_{s}\notag\\
    &=2\epsilon+k\vco^\da_k+\sum_{s< k,s\in\pos}(\nextp{s}-1)(\vco^\da_{s}-\vco^\da_{\nextp{s}})\notag\\
    &\leq2\epsilon+(1+\epsilon)k\vco^\da_k+\sum_{s< k,s\in\pos}(1+\epsilon)s\cdot(\vco^\da_{s}-\vco^\da_{\nextp{s}})\notag\\
    &=2\epsilon+(1+\epsilon)\sum_{s\leq k,s\in\pos}(s-\prevp{s})\vco^\da_{s}\notag\\
    &\leq2\epsilon+(1+\epsilon)\sum_{s'\leq k}\vco^\da_{s'}\leq(1+3\epsilon)\topl{k}{\vco},\label{eqn:major:0}
\end{align}
where we use $\nextp{s}-1\leq (1+\epsilon)s$ for each $s\in\pos$, and the fact that $\vco^\da$ is non-increasing. The last inequality is due to $\topl{k}{\vco}\geq\vco^\da_1=1$.

For each $k\notin\pos$, by the iterative definition of $\pos$, it follows that $(1+\epsilon)(k-1)\geq(1+\epsilon)\prevp{k}>k$, thus $k>1+1/\epsilon$ and $\nextp{k}<(1+\epsilon)k+1\leq(1+2\epsilon)k$. Using \eqref{eqn:major:0} we have just proved and the fact that $\rexp,\,\vco^\da$ are both non-increasing, we have
\[
\topl{k}{\rexp}\leq\topl{\nextp{k}}{\rexp}\leq(1+3\epsilon)\topl{\nextp{k}}{\vco}\leq(1+11\epsilon)\topl{k}{\vco},
\] 
where we use $\nextp{k}\leq(1+2\epsilon)k$ in the last inequality. \cref{lemma:majorization} gives the desired result.\qed
\section{A Simpler Special Case}\label{section:inner:topl}

In this section, we consider a special case of \glb, where each inner norm $\innorm_i=\topp{k_i},\,k_i\in[n]$ (recall for $\boldu\geq0$, $\topp{k_i}(\boldu)$ returns the sum of largest $k_i$ entries in $\boldu$), and the outer norm $\outnorm=\mcal{L}_\infty$. We call it the \maxtopk problem.

For any fixed $\epsilon>0$, we devise a deterministic $(3+\epsilon)$-approximation for \maxtopk. Analogous to previous algorithms, we set out to guessing the values at different indexes of the optimal assignment, but with slight modifications. 

Fix a small $\epsilon>0$. Define $\pos\subseteq[n]$ iteratively as in the PTAS in \cref{theorem:norm:value} for \normval, as well as $\mathsf{next},\,\mathsf{prev}$. Instead of optimizing $\topp{k_i}$ norm for machine $i$, we consider $\topp{k_i'}$ norm, where $k_i'=k_i$ if $k_i\in\pos$ and $k_i'=\prevp{k_i}$ otherwise. Notice by the iterative definition of $\pos$, if $k_i\notin\pos$, one has $k_i'<k_i<(1+\epsilon)\prevp{k_i}=(1+\epsilon)k_i'$.

Likewise, let $\sigma^\star:\jobs\rightarrow\machines$ be an unknown optimal assignment with optimum $\opt\geq0$, and $\vco\in\Rpos^n$ be defined in a subtly different way: for each $k\in[n]$, let $\vco_k$ be the maximum $k$-th largest job size among machines that has $k_i'\geq k$, that is, $\vco_k=\max_{i\in\machines:k_i'\geq k}\boldp_i\assigned{\sigma^\star}^\da_k$ (if there are no such machines or jobs, it is 0). It is easy to see that, since the assigned job-size vectors are non-increasingly sorted and the sets $\{i:k_i'\geq k\}$ are (inclusion-wise) non-increasing in $k$, $\vco$ is a non-increasing vector.
We guess $P=\{\rho_k\}_{k\in\pos}$ as a non-increasing vector of integer powers of $1+\epsilon$ such that 
\begin{enumerate}[label=(\roman*)]
	\item $\rho_k\in[\vco_k,(1+\epsilon)\vco_k)$ if $\vco_k\geq\epsilon\vco_1/n$,
	\item $\rho_k=(1+\epsilon)^{\ceil{\log_{1+\epsilon}(\epsilon\vco_1/n)}}$ otherwise.
\end{enumerate}

As before, after we guess the \emph{exact} value of $\vco_1$ and fix it, there are at most $(n/\epsilon)^{O(1/\epsilon)}$ such vectors. 
Suppose $\vco_1>0$ and all other guesses are correct in the sequel.
It follows that $\rho_1\leq (1+\epsilon)\vco_1\leq (1+\epsilon)\opt$ and for each $i\in\machines$,
\begin{enumerate}[label=(\roman*)]
    \item if $\vco_{k_i'}\geq\epsilon\vco_1/n$, we have $k_i'\rho_{k_i'}\leq (1+\epsilon)k_i'\vco_{k_i'}$. Because in the optimal solution, there exists a machine $i'\in\machines$ such that $k_{i'}'\geq k_i'$ and its $k_i'$-th largest assigned job size is $\vco_{k_i'}$, $i'$ has a $\topp{k_i'}$ norm at least $k_i'\vco_{k_i'}$. Combined with $k_{i'}\geq k_{i'}'\geq k_i'$, one has $k_i'\rho_{k_i'}\leq (1+\epsilon)\topl{k_i'}{\boldp_{i'}\assigned{\sigma^{\star}}}\leq
    (1+\epsilon)\topl{k_{i'}}{\boldp_{i'}\assigned{\sigma^{\star}}}\leq
    (1+\epsilon)\opt$;
    \item otherwise, we have $\rho_{k_i'}=(1+\epsilon)^{\ceil{\log_{1+\epsilon}(\epsilon\vco_1/n)}}\leq 2\epsilon\vco_1/n$, thus $k_i'\rho_{k_i'}\leq2\epsilon\vco_1\leq2\epsilon\cdot \opt$.
\end{enumerate}

Consider the following relaxation where $x_{ij}\geq0$ represents the extent we assign job $j$ to machine $i$.
\begin{align}
	\text{min} && r\tag{$\operatorname{M-Top}$}\label{lp:max:topl}\\
	\text{s.t.} && \sum_{j\in\jobs}(p_{ij}-\rho_{k_i'})^+x_{ij}&\leq r\quad\forall i\in\machines\tag{$\operatorname{M-Top}.1$}\label{lp:max:topl1}\\
	&& \sum_{i\in\machines}x_{ij}&=1\quad\forall j\in\jobs\tag{$\operatorname{M-Top}.2$}\label{lp:max:topl2}\\
	&& x_{ij}&=0\quad p_{ij}>\rho_1\tag{$\operatorname{M-Top}.3$}\label{lp:max:topl3}\\
	&& x&\geq0.\notag
\end{align}
\begin{lemma}\label{lemma:lp:maxtopl}
\ref{lp:max:topl} has optimum at most $\opt$.
\end{lemma}
\begin{proof}
    Define an integral solution $x^\star\in\{0,1\}^{\machines\times\jobs}$ according to the optimal assignment $\sigma^\star$ and let $r^\star=\opt$. 
    It suffices to show that $(x^\star,r^\star)$  satisfies \eqref{lp:max:topl1}.
    For each $i\in\machines$, because $\rho_{k_i'}\geq\vco_{k_i'}$, we have
    \begin{align*}
    \sum_{j\in\jobs}(p_{ij}-\rho_{k_i'})^+x_{ij}^\star=\sum_{j\in\sigma^{\star-1}(i)}(p_{ij}-\rho_{k_i'})^+\leq\sum_{j\in\sigma^{\star-1}(i)}(p_{ij}-\vco_{k_i'})^+.
    \end{align*}
    
    Since the $k_i'$-th largest job assigned to $i$ has size at most $\vco_{k_i'}$ by definition, the above sum has at most $k_i'$ non-zero entries, thus at most $\topl{k_i'}{\boldp_i\assigned{\sigma^\star}}$. The lemma now follows since $k_i'\leq k_i$ and $\topl{k_i'}{\boldp_i\assigned{\sigma^\star}}\leq\topl{k_i}{\boldp_i\assigned{\sigma^\star}}\leq \opt$.
\end{proof}

We solve \ref{lp:max:topl} and obtain a solution $(\bar x,\bar r)$, and $\bar r\leq \opt$ by \cref{lemma:lp:maxtopl}. Using the classic rounding algorithm by Shmoys and Tardos \cite{shmoys1993approximation} (based on the original job sizes $p_{ij}$), we obtain an integral assignment $\hat x\in\{0,1\}^{\machines\times\jobs}$ such that $\sum_i\hat x_{ij}=1$ for each job $j$. 
More precisely, make $n_i:=\ceil{\sum_j\bar x_{ij}}$ copies of each machine $i$; in non-increasing order of $p_{ij}$, fractionally assign the jobs $\{j:\bar x_{ij}>0\}$ to the same extent, and sequentially on the copies of $i$, such that the first $n_i-1$ copies are all assigned exactly unit mass of jobs. We then use standard methods \cite{lovasz2009matching} and round the resulting fractional matching to an arbitrary integral matching $\hat x$, and each job is matched due to \eqref{lp:max:topl2}.

Since we assign the jobs in non-increasing order of $p_{ij}$ to the copies, it follows that for the $t$-th copy $i_t$ of machine $i$, $t\geq2$, the assigned job size under $\hat x$ is at most the \emph{average} job size on $i_{t-1}$ under $\bar x$. Each assigned job size is at most $\rho_1$ due to \eqref{lp:max:topl3}.
Hence for each machine $i$, the $\topp{k_i'}$ norm must be attained over its $k_i'$ foremost copies $\{i_1,\dots,i_{k_i'}\}$, which is at most $\rho_1+\sum_{t\leq k_i'}\sum_jp_{i_tj}\bar x_{i_tj}$.
Using simple subadditivity $(s+t)^+\leq s^++t^+$ and that each machine copy is assigned to an extent of at most 1, this is in turn bounded by
\begin{align*}
\rho_1+k_i'\rho_{k_i'}+\sum_{t\leq k_i'}\sum_{j}(p_{i_tj}-\rho_{k_i'})^+\bar x_{i_tj}
&\leq \rho_1+k_i'\rho_{k_i'}+\sum_{j\in\jobs}(p_{ij}-\rho_{k_i'})^+\bar x_{ij}\\
&\leq (2+2\epsilon)\opt+\bar r\leq (3+2\epsilon)\opt.
\end{align*}

Finally using $k_i\leq (1+\epsilon)k_i'$, the $\topp{k_i}$ norm is bounded by $(1+\epsilon)(3+2\epsilon)\opt\leq(3+7\epsilon)\opt$, whence we obtain the following theorem.
\begin{theorem}\label{theorem:maxoftopl}
For each $\epsilon>0$, there exists a deterministic $(3+\epsilon)$-approximation algorithm for \emph{\maxtopk} with running time $(n/\epsilon)^{O(1/\epsilon)}\cdot poly(m,n)$.
\end{theorem}

\end{document}